\documentclass[10pt,journal]{IEEEtran}
\pdfoutput=1
\normalsize
\usepackage{lipsum,amsthm,amssymb,amsmath,tikz,graphics,cite,float,graphicx,epstopdf,epsfig,verbatim,url}
\usepackage{algorithm,algpseudocode}
\usepackage{color,soul}
\IEEEoverridecommandlockouts
\IEEEpeerreviewmaketitle
\usetikzlibrary{automata}
\usetikzlibrary{shapes,arrows}
\newtheorem{lem}{Lemma}
\usepackage[utf8]{inputenc}
\newtheorem{thm}{Theorem}
\newtheorem{rem}{Remark}

\newtheorem{mydef}{Definition}
\newcommand{\BigO}[1]{\ensuremath{\operatorname{O}\bigl(#1\bigr)}}
\tikzset{
  treenode/.style = {align=center, inner sep=0pt, text centered,
    font=\sffamily},
  arn_r/.style = {treenode, circle, black, draw=black, 
    text width=1.5em, very thick},
}

\begin{document}
\title{Multihop Caching-Aided Coded Multicasting for the Next Generation of Cellular Networks}

\author{Mohsen~Karimzadeh~Kiskani$^{\dag}$,
        and Hamid~R.~Sadjadpour$^{\dag}$,
\thanks{Copyright (c) 2015 IEEE. Personal use of this material is permitted. However, permission
to use this material for any other purposes must be obtained from the IEEE by sending a request to
pubs-permissions@ieee.org.}
\thanks{M. K. Kiskani$^{\dag}$ and H. R. Sadjadpour$^{\dag}$ 
are with the Department of Electrical Engineering, University of California, Santa Cruz. Email: 
\{mohsen, hamid\}@soe.ucsc.edu}}
\maketitle 
\begin{abstract}
Next generation of cellular networks deploying wireless distributed femtocaching
infrastructure proposed by Golrezaei et. al. 
{are}
studied. By taking advantage of multihop communications in each cell,  
the number of required {\em {femtocaching} helpers} is significantly reduced. This reduction  
{is}
achieved by using 
{ 
the underutilized storage and communication capabilities in 
the User Terminals (UTs),}
which results in reducing the deployment costs of distributed femtocaches. 

A multihop index coding technique is proposed 
to code the cached contents in helpers to achieve order optimal capacity gains. 
As an example, we consider a wireless cellular system in which contents have a popularity distribution and demonstrate that  
our approach can replace many unicast 
communications with multicast communication. 
We will prove that simple heuristic linear index code algorithms based on graph coloring  can achieve  order
optimal capacity under Zipfian  content popularity distribution. 
\end{abstract}
\markboth{IEEE Transactions on Vehicular Technologies}{IEEE Transactions on Vehicular Technologies}
\IEEEpeerreviewmaketitle

\section{Introduction}
With the recent pervasive surge in using wireless devices for video and high speed data transfer, it seems eminent that
the current 
wireless cellular networks cannot be a robust solution to the ever-increasing wireless bandwidth utilization demand. Researchers
have  been recently
focused on laying down the fundamental grounds for future cellular networks to overcome such problems.

Deploying home size base
stations is proposed as a 
solution in \cite{chandrasekhar2008femtocell}. 
The solution in \cite{chandrasekhar2008femtocell} is based on the idea of 
femtocells in which many small cells are deployed throughout the
network to cover the entire network. Deploying 
many small cells in the network with reliable backhaul links requires significant capital investment.
Therefore,
Golrezaei et. al.  \cite{golrezaei2012femtocaching} proposed femtocaching as 
an alterante solution to overcome this problem. 
In this approach, in every cell along with the main base station, smaller 
base stations with low-bandwidth backhaul links and high storage capabilities are deployed  to create a wireless 
distributed caching infrastructure. These small base stations which are called caching helpers (or simply helpers),
will store popular contents in their caches and use their caches to serve \emph{User Terminal (UT)}
requests. Therefore, in
networks with
high content reuse, the backhaul utilization will be significantly reduced using this approach. If the requested content 
is not available in the 
helper's cache, UTs can still download the content from their low-bandwidth backhaul links to the base station. 
The proposed technique in \cite{golrezaei2012femtocaching} requires deployment of large number of femtocaches in order to cover all the nodes in the network. 


On the other hand, it is well-known that web content request popularity follows 
Zipfian-like distributions \cite{breslau1999web}. 
This content popularity distribution implies that few popular contents are widely requested by the network UTs. 
We assume UTs store their requested  contents and therefore, helpers can multicast multiple  requests by taking advantage of coding to reduce the total number
of transmissions.

In this paper, we propose to use index coding to code the contents in  helpers
before transmission. Index coding is a source coding technique 
proposed in \cite{bar2011index} which takes advantage of UTs' side information
in broadcast channels to minimize 
the required number of transmissions.  In index coding, the source (e.g. base station) 
designs codes based on the side information stored in requesting nodes. The coded information is  broadcasted  to the UTs that use the information together with
their 
cached contents to decode the desired content.
Index coding improves bandwidth utilization by minimizing the number of 
required transmissions. We propose to extend index coding 
approach from broadcast one-hop communication to multihop  scenarios
which will be explained in 
details.

Our main motivation to use index coding is the high storage availability in UTs to improve the achievable throughput of the future wireless cellular networks. 
Current improvements in high density storage systems 
has made it possible to have personal devices with 
Terabytes of storage capability. This ever-increasing trend provides future personal wireless devices with huge under-utilized 
storage capabilities. Future wireless devices can use their storage capability to store the contents that they 
have already requested. In an index coding setting, many UTs that are requesting different contents can receive a coded file which 
is multicasted to them and then each UT  uses the information in its cache to decode its requested content from the received 
coded file. 
There is an 
important equivalence between index coding and network coding as stated in \cite{el2010index, effros2012equivalence} 
and therefore, the results in this paper can also be stated based on a network coding terminology.

We will prove that index coding can be efficiently used to encode the contents by helpers 
under a Zipfian  distribution model. The encoded contents can be relayed through
multiple hops to all the UTs being
served by that helper.

The optimal index coding solution is an NP-Hard problem  \cite{langberg2011hardness}. However, we will show 
that even using 
linear index codes can result {in} order optimal capacity gains in these networks. We believe that this coding technique
can serve as a 
complement to the solution proposed in \cite{golrezaei2012femtocaching}. 
As clearly articulated in \cite{maddah2014fundamental}, in any caching problem we are faced with two  phases 
of  cache placement and cache 
delivery. While \cite{golrezaei2012femtocaching} proposes efficient cache placement algorithms, we
will be focusing on efficient 
delivery methods for their solution. We will show that the problem of delivery in \cite{golrezaei2012femtocaching} can be efficiently 
addressed by using index coding in 
the helpers.

Recent discussions on standards for future 5G cellular networks are focused on providing 
 high bandwidth for Device-to-Device communications (D2D). Examples of such approaches are the  
 IEEE 802.11ad standard (up to 60GHz \cite{ieee80211ad}) and the millimeter-wave proposal 
which can potentially enable up to 300GHz of bandwidth for D2D communications
\cite{boccardi2014five,mmrange}. This potential abundant D2D
bandwidth can be  utilized to relay the coded contents inside an ad hoc network which is being served by a
helper. It is indeed such excessive storage and bandwidth capabilities of future wireless systems that make our solution 
feasible. 

Deploying many femtocaching helpers is not economically efficent. 
On the
other hand, since UTs have significant D2D capabilities, they can 
efficiently participate in content delivery through multihop D2D 
communications. In our proposal, we suggest to deploy few helpers 
which can deliver the contents to neighboring UTs through multihop 
D2D communications. This can reduce the network deployment and 
maintenance costs.

The rest of this paper is organized as follows. Section \ref{relwork} reviews the related works and 
section \ref{netmodel} describes the proposed network model that is similar to \cite{golrezaei2012femtocaching} with the addition of using multihop communications and 
index coding. In section 
\ref{mohsensec}, we will explain the scaling laws of capacity improvement  using index coding and relaying. Section \ref{graphcoloringsec} demonstrates  
that index coding algorithms can achieve order optimal gains. Section \ref{discuss} describes 
the simulation results and the paper is concluded in section \ref{conclude}.

\section{Related work}
\label{relwork}
The problem of 
caching when a server is transmitting contents to clients was studied in \cite{maddah2014fundamental} from an 
information theoretic point of view. The authors introduced 
two phases of {\em cache placement} and {\em content delivery}.  
For a femtocaching solution, efficient cache placement
 algorithms are proposed 
 in \cite{golrezaei2012femtocaching}. 
  In this paper, we focus on efficient content delivery algorithms 
 through index coding and 
 multihop D2D communications for a femtocaching solution.

There has been significant research on index coding since it was proposed 
in \cite{bar2011index}. The practical implementation of index coding for 
wireless applications was proposed in  \cite{chaudhry2011complementary} 
by proposing cycle counting methods.
A dynamic index coding solution for wireless
broadcast channels is  proposed in \cite{neely2013dynamic}. 
In \cite{neely2013dynamic}, a wireless broadcast station is considered and a 
simple set of codes based on cycles in the
dependency graph is provided. They show the optimality of these codes for a
class of broadcast relay problems. In this paper, we prove 
that codes based on cycles can acheive order optimal capacity gains in 
networks with Zipfian content request distribution.

Approximating index coding solution is proved \cite{bar2011index} to be  an NP-Hard problem. However, 
efficient heuristics has been proposed in \cite{chaudhry2008efficient} which are based on well-known graph coloring
algorithms. 
Other references like 
\cite{el2010index} and \cite{effros2012equivalence} have studied the 
connections and equivalence between index coding and network coding. 
Tran et al. \cite{tran2009hybrid}  studied a single hop wireless link 
from a network coding approach and showed similar results to 
index coding.
Ji et al. \cite{ji2014fundamental} studied theoretical limits of caching in 
D2D communication networks.

Study of coding techniques in networks with high content reuse has recently attracted the attention of researchers. 
Montpetit et al. \cite{montpetit2012network}  studied the applications of network coding in Information-Centric Networks (ICN). 
Wu et al. \cite{wu2013codingcache}  studied network coding in Content Centric Networks (CCN) which is an implementation of ICN. 
Leong et al. \cite{leong2009optimal}  proposed a linear programming formulation to deliver contents optimally in today's IP based Content Delivery Networks (CDN) using network coding.
Llorca et at. \cite{llorca2013network} proposed 
a network-coded caching-aided multicasting technique for efficient content delivery in CDNs. This paper extends Index Coding to multihop communication and
 shows that  linear codes 
can achieve order 
optimal capacity gains in such networks. 


\section{Network Model}
\label{netmodel}
 We assume a network model for future wireless cellular
 networks  in which femtocaching helpers with significant storage capabilities 
 are deployed throughout the network to assist the efficient delivery of contents to 
UTs through reliable D2D communications. Such helpers are 
characterized with low rate backhaul links which can be wired or wireless.
They will also have localized, high-bandwidth communication capabilities. 
With their significant storage capacity, 
in a network with high content reuse, many of the content requests for the UTs inside a cell can be satisfied directly 
by the helpers. 

Each helper is serving a wireless ad-hoc network in which the UTs are utilizing  high bandwidth 
D2D communication techniques such as  millimeter wave and IEEE 802.11ad technologies.  
This high bandwidth D2D communication enables 
the UTs to relay  data from a nearby helper to all the UTs that are within transmission range. We assume that 
the path lifetime between the helper and UTs is longer than the time required to transmit the content. For  
large files and when UTs are moving fast, one solution is to divide the content into  smaller files and treat 
each file separately. 
The paper assumes that the network is connected which can be justified by 
the large number of UTs that will be available in the future wireless 
cellular networks.

%


Further, we take advantage of index coding and the side information cached in UTs 
to significantly reduce the required number of 
helpers and consequently, reduce the infrastructure maintenance and deployment costs. To demonstrate the effectiveness of 
using multihop communications, we consider similar assumptions as  \cite{golrezaei2012femtocaching} with 
a macro base station placed in the center of a cell with radius 400 meters 
serving 1000 UTs and a transmission range of 100 meters \cite{mmrange} for D2D 
communication. As shown in 
Figure \ref{fig_multihop},
with only 
{4}
helpers uniformly located in the cell, 100\% and 80\% of nodes are covered
with 3 and 2 hop communications respectively. Covering all the UTs in the
same cell with only one hop communication requires up to 27 helpers
\cite{golrezaei2012femtocaching}.
\begin{figure}[http]
    \center
      \includegraphics[scale=0.5,angle=0]{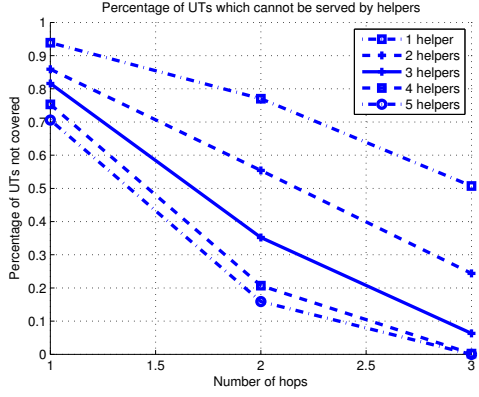}
      \caption{{Pertentage of UTs not covered versus the maximum number of hops traveled. This 
      simulation is carried over a cell with radius 400 meters and with a communication range of 
       100 meters.}}
 \label{fig_multihop}
\end{figure}
Clearly, for a vehicle or a mobile UT operating in the cell, the handover
probability will be 
significantly reduced if the number of helpers shrinks from 27 to 4.

We assume that $n$ UTs denoted by $\mathbb{N}=\{N_1,N_2,...,N_n\}$ are being served by a helper. There are 
 $m$ contents $\mathbb{M}=\{M_1,M_2,...,M_m\}$  available  with $M_1$ as the most popular content  and $M_m$ as the least popular content in the network .
Let's assume UT $N_i$ requests a content with popularity index $r_i$ in the current time interval. 
Each UT has a cache of fixed size $\delta$ in which contents with indices 
$C_i=\{{c_{i1}},...,{c_{i\delta}}\}$ are stored. Therefore, we assume that UT $N_i$ 
caches contents $M_{{c_{i1}}},M_{{c_{i2}}},\dots,M_{{c_{i\delta}}}$ and the set of  cached 
content indices in $N_i$ is represented by $C_i$. Therefore, if we denote 
the set of cached contents in UT $N_i$ by $\mathcal{M}_i$, then we have 
$\mathcal{M}_i = \{M_{j}~|~ j \in C_i\}$.
The requested content index is shown by 
$r_i$.

Let's assume that we have $n$ UTs each with a 
set of side information $\mathcal{M}_i$. The formal definition of index code described in  \cite{bar2011index} is given below.
\begin{mydef}{\em 
 An {\em index code} on a set of $n$ UTs each with a side information 
set $\mathcal{M}_i \subseteq \mathbb{M}$, and a requesting content 
${M}_{r_i} \in \mathbb{M}$ for $i=1,2,\dots,n$ is 
defined as a set of codewords 
in $\{0,1\}^l$ toghether with 
\begin{enumerate}
 \item An encoding function $\mathcal{E}$ mapping inputs 
 in $\{0,1\}^m$ to codewords and 
 \item A set of decoding functions $\mathcal{D}_1,\mathcal{D}_2,\dots,\mathcal{D}_n$
 such that $\mathcal{D}_i (\mathcal{E}(\mathbb{M}), \mathcal{M}_i) 
 = {M}_{r_i}$ for $i=1,2,\dots,n$.
\end{enumerate}
}\label{indx_def} 
\end{mydef}
In the above definition, the length of the index code $l$ denotes the 
number of required transmissions to satisfy all the content requests of 
the $n$ UTs. The encoding function $\mathcal{E}$ is applied to the 
contents by the helper and the decoding functions $\mathcal{D}_i$s  
are applied individually by UTs to decode their desired contents from 
the encoded content using their cached information.

Figure \ref{fig_ex_netwoork} shows a helper $H$ serving 6 UTs $N_1, N_2, N_3, N_4, N_5,$ 
and $N_6$. 
Let's assume UTs $N_1, N_2$ and $N_5$  request  contents 
$M_3, M_1$ and $M_4$ while storing $\{M_1,M_4\}$, $\{M_3,M_4\}$ and 
$\{M_1,M_3\}$ respectively. Using 
index coding requires 3 channel usages while without index coding, we need 5 channel usages. 
This is true since using index coding, the helper creates the XOR 
combination of contents $M_1, M_3$ and $M_4$ as $M_1 \oplus M_3 \oplus M_4$ and broadcasts
this coded content to its neighboring UTs. Either of $N_1$ and $N_2$ can immediately
reconstruct their requested content from this coded content. 
For instance, $N_1$ which is requesting $M_3$ can decode its requested content by using
it's cached information and XOR operation on the encoded message to retrieve the 
requested content $M_3$, i.e.,  $(M_1 \oplus M_3 \oplus M_4) \oplus M_1 \oplus M_4 = M_3$.
After $N_2$ receives $M_1 \oplus M_3 \oplus M_4$, it relays it to  node $N_4$ which 
 simply broadcasts it to $N_5$. UT $N_5$ will again use XOR operation to decode it's 
desired content $M_4$ by adding it's cached contents to the coded content that it has received, 
i.e.,  $(M_1 \oplus M_3 \oplus M_4) \oplus M_1 \oplus M_3 = M_4$. Therefore, 
only 3 
transmissions are needed. 
 
To satisfy the content requests through multihop D2D  without 
index coding, 5 transmissions are needed as the helper 
should transmit $M_3$ to $N_1$, $M_1$ to $N_2$, 
$M_4$ to $N_2$, and $N_2$ relays 
$M_4$ to $N_4$ and $N_4$ relays it to $N_5$.
\begin{figure}
\centering
 \begin{tikzpicture}[->,>=stealth',shorten >=1pt,auto,node distance=1cm, semithick]
  \node[rectangle,draw]  (H) {H};
  \node[circle,draw]  (A0) [right of=H] {};
  \node[circle,draw]  (A1) [above right of=H] {};
  \node[circle,draw]  (A2) [below right of=H] {};
  \node[circle,draw]  (A3) [right of=A0] {};
  \node[circle,draw]  (A4) [above right of=A3] {};
  \node[circle,draw]  (A5) [below right of=A3] {};
  
  \node (T1) [right of = A1, xshift=-0.6cm] {\footnotesize{$N_1$}};
  \node (T2) [right of = A2, xshift=-0.6cm] {\footnotesize{$N_3$}};
  \node (T3) [above of = A0, yshift=-0.7cm, xshift=0.2cm] {\footnotesize{$N_2$}};
  \node (T4) [above of = A3, yshift=-0.6cm] {\footnotesize{$N_4$}};
  \node (T5) [right of = A4, xshift=-0.6cm] {\footnotesize{$N_5$}};
  \node (T6) [right of = A5, xshift=-0.6cm] {\footnotesize{$N_6$}};
  
  \path
  (H) edge node {} (A0);
  \path
  (H) edge node {} (A1);
  \path
  (H) edge node {} (A2);
  \path (A0) edge node {} (A3);
  \path (A3) edge node {} (A4);
  \path (A3) edge node {} (A5);
\end{tikzpicture}
  \caption{Example of a wireless multihop network being served by the helper $H$. Each arrow represents a  link with high bandwidth 
  D2D communication capability.}
  \label{fig_ex_netwoork}
\end{figure}
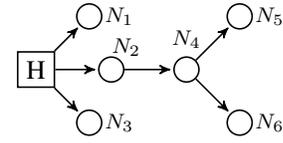

We assume a Zipfian distribution with parameter $s > 1$ for content popularity distribution in the network. This means that the 
probability that UT $N_i$ requests any content with index $r_i$ at any time instant is given by 
\begin{equation}
 \label{ebsdf}
 \textrm{Pr}[N_i ~\textrm{requests content with index}~ r_i] = \frac{r_i^{-s}}{H_{m,s}},
\end{equation}
where $H_{m,s} = \sum_{j=1}^{m} \frac{1}{j^s}$ denotes the $m^{th}$ generalized harmonic number with parameter $s$. 

\begin{rem} {\em 
This paper only focuses on Zipfian content request 
 probability with parameter $s>1$ which is a {\em non-heavy-tailed} probability distribution. 
 Our results are correct for any type of non-heavy-tailed  probability distribution and 
 the extension to heavy-tailed probability distribution is the subject of future work.}
\label{rem_prob_explain}
\end{rem}

Dependency graph  is a useful analytical tool \cite{bar2011index, el2010index, chaudhry2011complementary} that is widely used in index coding literature.
\begin{mydef}{\em
 \emph{(Dependency Graph):} Given an instance of an index coding problem, the dependency graph\footnote{Dependency graph 
 is a directed graph.} 
 $\overrightarrow{G}(V,E)$ is defined as 
 \begin{itemize}
  \item each UT $N_i$ corresponds to a vertex  in $V$, $N_i \in V$, and 
  \item there is a directed edge in $E$ from $N_i$ to $N_j$ if and only if $N_i$ is requesting a 
  content that is already cached in $N_j$. 
 \end{itemize}
 } \label{depgraph}
\end{mydef}
This dependency graph does not represent the actual physical links between UTs in the network. This is a virtual graph in which
each edge represents the connection between a UT that is requesting a content and a UT that caches this content. 
As discussed in \cite{chaudhry2011complementary,neely2013dynamic}, every cycle in the dependency 
graph is representative of a connection between UTs and it can save one transmission. For every clique
in the dependency graph, all the requesting UTs in the clique can be satisfied by a simple linear XOR index code. 
The complement of dependency graph is called \emph{conflict graph}. 
This graph is  of significant interest since any clique in the dependency graph gives rise to an independent 
set\footnote{An independent set is a set of vertices in a graph for which none of the vertices are connected by an edge.}
in the conflict graph. Therefore, well-known graph coloring algorithms over conflict graph can be used to find simple linear 
{XOR}
index codes.
The dependency and conflict graphs in our network are random directed graphs. 
In the next section, we will use the properties of these graphs 
to find the capacity gains and propose simple index coding solutions.  

To prove our results in this paper, we have used Least Recently Used (LRU) 
 or Least Frequently Used (LFU) cache  policies. Similar results can be 
produced for other caching policies. LRU caching policy  assumes that 
most recently requested contents 
are kept in the cache. 
In LRU caching, in each time slot the content that is requested in the previous time slot is 
stored in the first cache location. If this content was already available in the cache, it is 
moved from that location to the first cache location. If the content was not 
available, then the content that is least recently used is discarded and the most recently
requested content is cached in the first cache location. Any other content is relocated to a 
new cache location such that the contents appear in the order that they have been requested.

In LFU caching policy, the contents are cached based on their request frequency. Highly popular 
contents are stored in the first locations of the cache and contents with lower request frequency are 
cached in the bottom locations of the cache. 

Computing the probability of having content with index $r_{i}$ in cache $C_{j}$, $\textrm{Pr}[r_{i} \in C_{j}]$,
turns out to be  complicated for LRU { or LFU} caching policies. 
A  simple lower bound on this probability for LRU
can be found by noticing that 
$\textrm{Pr}[r_{i} \in C_{j}]$ is 
greater than the probability that UT $N_j$ have requested the content with index $r_{i}$ in the most recent time slot
and therefore it is located at the top of the  
cache. This lower bound can be derived using equation \eqref{ebsdf}.
\begin{align}
 \textrm{Pr}[r_i \in C_j] &= 
  \textrm{Pr}[c_{j1} = r_i] + \sum_{l=2}^{\delta} \textrm{Pr}[c_{jl} = r_i]          \nonumber \\
 &\ge \textrm{Pr}[c_{j1} = r_i] = \frac{r_i^{-s}}{H_{m,s}}. 
 \label{validating_eq}
\end{align}
To prove that the same lower bound holds for LFU, notice that $\textrm{Pr}[r_{i} \in C_{j}]$ 
is larger than the same probability when the cache size is $M=1$. Therefore, 
\begin{equation}
 \textrm{Pr}[r_i \in C_j]  \ge \textrm{Pr}[r_i \in C_j~|~M=1] 
 = \textrm{Pr}[c_{j1} = r_i]   = \frac{r_i^{-s}}{H_{m,s}}. 
 \label{validating_eq2}
\end{equation}
In the subsequent sections, we will use these lower bounds to prove our results. 

In this paper, we will state our results in terms of order bounds. To avoid any confusion, we use the following order notations \cite{knuth1976big}.
We denote $f(n)=\operatorname{O}(g(n))$ if there exist $c>0$ and $n_0>0$ such that $f(n)\leq c g(n)$ for all $n\geq n_0$, $f(n)=\Omega(g(n))$ if $g(n)=\operatorname{O}(f(n))$, and $f(n)=\Theta(g(n))$ if $f(n)=\operatorname{O}(g(n))$ and $g(n)=\operatorname{O}(f(n))$.

\section{Order optimal capacity gain}
\label{mohsensec}

In this section, we will prove that index coding can significantly 
decrease the number of transmissions
in the network. We will 
specifically use the Zipfian 
content distribution in the underlying content distribution network. To do so, 
we will first state and prove the following  lemma. 
\begin{lem}
 \label{hproblem}{\em
 Let's consider a Zipfian content distribution with parameter $s > 1$ and parameter $h_{\epsilon} = {\epsilon}^{\frac{1}{1-s}}$ where $0<  \epsilon < 1$.
For 
 every $i$ with popularity index $r_i$ that is less than $h_{\epsilon}$, the request probability is at least $1-\epsilon$.
}\end{lem}
\begin{proof}
 \label{prroof}
Based on the Zipfian distribution assumption and equation \eqref{ebsdf}, 
{ the probability that the requested content has a popularity of at most 
$h_{\epsilon}$ is equal to 
\begin{align}
 \textrm{Pr}[r_j \le h_{\epsilon}] = \frac{H_{h_{\epsilon},s}}{H_{m,s}}. 
 \label{request_prob}
\end{align}
In order to satisfy $\textrm{Pr}[r_j \le h_{\epsilon}] \ge 1 - \epsilon$, we should 
have 
\begin{align}
 \epsilon &\ge 1- \textrm{Pr}[r_j \le h_{\epsilon}] 
 = 1- \frac{H_{h_{\epsilon},s}}{H_{m,s}} 
 = \dfrac{1}{H_{m,s}} \sum_{j=h_{\epsilon}+1}^m j^{-s} \nonumber \\
 &= \dfrac{1}{H_{m,s}} \sum_{i=1}^{\lfloor \frac{m}{h_{\epsilon} }\rfloor -1 } 
 \sum_{j=i h_{\epsilon}+1}^{(i+1)h_{\epsilon}} j^{-s}
 + \dfrac{1}{H_{m,s}} \sum_{j=\lfloor \frac{m}{h_{\epsilon} }\rfloor
 h_{\epsilon}+1}^{m} j^{-s}.
 \label{eqs_proof_lem1}
\end{align}
If we have 
\begin{align}
 \epsilon &\ge 
 \dfrac{1}{H_{m,s}} \sum_{i=1}^{\lfloor \frac{m}{h_{\epsilon} }\rfloor } 
 \sum_{j=i h_{\epsilon}+1}^{(i+1)h_{\epsilon}} j^{-s},
 \label{eqs_proof_lem2}
\end{align}
then the inequality in \eqref{eqs_proof_lem1} will certainly hold since the right 
hand side in \eqref{eqs_proof_lem2} is larger than the right hand side in 
\eqref{eqs_proof_lem1}. Note that, for 
$i h_{\epsilon}+1 \le j \le (i+1)h_{\epsilon}$ we have 
$j^{-s} < (i h_{\epsilon})^{-s}$. Hence, 
\begin{align}
 \sum_{j=i h_{\epsilon}+1}^{(i+1)h_{\epsilon}} j^{-s} <  
 \sum_{j=i h_{\epsilon}+1}^{(i+1)h_{\epsilon}} (i h_{\epsilon})^{-s} =
 h_{\epsilon} (i h_{\epsilon})^{-s} = i^{-s} h_{\epsilon}^{1-s} 
 \label{eqs_proof_lem3}
\end{align}
This means that if 
\begin{align}
 \epsilon &\ge 
h_{\epsilon}^{1-s}
\dfrac{1}{H_{m,s}} \sum_{i=1}^{\lfloor \frac{m}{h_{\epsilon} }\rfloor } i^{-s},
 \label{eqs_proof_lem4}
\end{align}
then \eqref{eqs_proof_lem2} holds since the right hand side of inequality  \eqref{eqs_proof_lem2}
is smaller than the right hand side of inequality  \eqref{eqs_proof_lem4}
as shown by  \eqref{eqs_proof_lem3}. 
Notice that since,
\begin{align}
\dfrac{1 }{H_{m,s}}\sum_{i=1}^{\lfloor \frac{m}{h_{\epsilon} }\rfloor } i^{-s} 
=\dfrac{\sum_{i=1}^{\lfloor \frac{m}{h_{\epsilon} }\rfloor } i^{-s}}
{\sum_{i=1}^{\lfloor \frac{m}{h_{\epsilon} }\rfloor } i^{-s} 
+\sum_{\lfloor \frac{m}{h_{\epsilon} }\rfloor+1 }^m i^{-s}} < 1,
 \nonumber
\end{align}
if $h_{\epsilon}$ is chosen such that 
\begin{align}
 \epsilon &\ge h_{\epsilon}^{1-s},
 \label{eqs_proof_lem6}
\end{align}
then all of the inequalities in
equations \eqref{eqs_proof_lem1}, \eqref{eqs_proof_lem2},
\eqref{eqs_proof_lem3} and \eqref{eqs_proof_lem4} 
will be valid and hence $\textrm{Pr}[r_j \le h_{\epsilon}] \ge 1 - \epsilon$.
Therefore, in order to have 
$\textrm{Pr}[r_j \le h_{\epsilon}] \ge 1 - \epsilon$, it is enough to 
choose $h_{\epsilon}$ such that \eqref{eqs_proof_lem6} is valid. Hence, if 
$h_{\epsilon}$ is chosen to be at least equal to 
\begin{align}
  h_{\epsilon} = {\epsilon}^{\frac{1}{1-s}},
 \label{eqs_proof_lem7}
\end{align}
then we have $\textrm{Pr}[r_j \le h_{\epsilon}] \ge 1 - \epsilon$.  Notice that 
the choice of $h_{\epsilon}$ in \eqref{eqs_proof_lem7} is such that it only
depends on $\epsilon$ and is independent of $m$.
}
\end{proof}
For instance for $s=2$ and $\epsilon=0.01$, $h_{\epsilon}$ can be chosen as 100 
(regardless of the size of $m$). This implies that for a Zipfian 
distribution with $s=2$, 100 highly popular contents among any large number of contents 
would account for 99\% of the total content requests.  Therefore, 
if $h_{\epsilon}$ is chosen as in equation \eqref{eqs_proof_lem7},
with a probability of at least $1-\epsilon$  all content
requests have popularity index of at most $h_{\epsilon}$. Now define
$p_{\epsilon}$ as
 \begin{equation}
 \label{edgeproblowerbound}
 p_{\epsilon} \triangleq \dfrac{h_{\epsilon}^{-s}}{H_{m,s}}
 = \dfrac{{\epsilon}^{-\frac{s}{1-s}}}{H_{m,s}}.
\end{equation}
Based on above discussion, in our instance of index coding dependency graph,
 with a probability of at least $1-\epsilon$ edges are present with a 
probability of at least $p_{\epsilon}$. We will discuss this in more details 
later in the proof for Theorem \ref{indexlimit}.  Notice that for large 
values of $m$, $p_{\epsilon}$ is also independent of $m$ since in that case 
$H_{m,s} \approx \zeta(s)$ and $p_{\epsilon}$ only depends on $\epsilon$ and $s$.

As stated in \cite{chaudhry2011complementary}, if we choose the right encoding vectors for any index coding problem,
for any vertex disjoint cycle in the dependency graph we can save one transmission. 
Therefore, 
the number of vertex-disjoint cycles\footnote{These are the cylces that do not have any common vertex.} 
in the dependency graph can serve as a lower bound for the number 
of saved transmissions in any index coding problem. 
Number of vertex disjoint cycles is also used in \cite{neely2013dynamic} as a way of finding the lower bound for
index coding gain. To count the number of vertex-disjoint cycles in our random dependency graph, we will use the 
following lemma originally proved as Theorem 1 in \cite{erdHos1962maximal}.
\begin{lem}
\label{lemerdos}
\emph{
Let $d > 1$ and $v \ge 24d$ be integers. Then any graph $G^{v}_{f(v,d)}$ with $v$ vertices and at least 
$f(v,d)= (2d-1)v - 2d^2+d$ edges   
 contains $d$ disjoint cycles or $2d-1$ vertices of degree $v-1$.
 \footnote{Clearly, this lemma is valid when the number of edges is more than $f(v,d)$.}
}
\end{lem}
Note that the dependency graph is a directed graph and in order to use 
Lemma \ref{lemerdos}, we need to construct 
an undirected graph.    
Let's denote the directed and undirected random graphs on $n$ vertices and edge presence probability $p_{\epsilon}$ by 
$\overrightarrow{G}(n,p_{\epsilon})$ and 
${G}(n,p_{\epsilon})$, respectively.
 In a directed graph $\overrightarrow{G}(n,p_{\epsilon})$, the  
probability that two vertices are connected by two opposite directed edges  is $p_{\epsilon}^2$. Therefore, we can build an undirected
graph ${G}(n,p_{\epsilon}^2)$ 
with the same number of vertices and an edge between two vertices if there are two opposite directed edges in the directed graph
$\overrightarrow{G}(n,p_{\epsilon})$ between these two UTs. Hence, $\overrightarrow{G}(n,p_{\epsilon})$
essentially contains a copy of $G(n,p_{\epsilon}^2)$. Note that there are some edges between UTs in  $\overrightarrow{G}(n,p_{\epsilon})$ that 
do not appear
in $G(n,p_{\epsilon}^2).$ This fact was also observed in \cite{haviv2012linear}.
%
%
%
Therefore, a lower bound on the number of disjoint cycles for  $G(n,p_{\epsilon}^2)$ 
 implies a lower bound on the number of disjoint cycles for $\overrightarrow{G}(n,p_{\epsilon})$.


In the following theorems, we will use Lemma \ref{lemerdos} to prove that using
index coding to code the contents 
can be very efficient. 
\begin{thm}
 \label{indexlimit}
 {\em 
 Assume all UTs are utilizing LRU or LFU caching policies for a Zipfian content
 request distribution with parameter $s>1$.
 Index coding can save 
 $\Omega(n p_{\epsilon}^2)$ transmissions for any helper serving $n$
 UTs with a probability of at least $1-\epsilon$ for any $0 < \epsilon < 1$.
}\end{thm}
\begin{proof}
 Consider a Zipfian distribution with parameter $s>1$ and
let $0<\epsilon<1$ be fixed.
The dependency graph $\overrightarrow{G}(V,E)$ in our problem
is composed of $n$ vertices $N_{1}, N_{2}, \dots, N_{n}$ which correspond to the $n$ UTs that are served by a helper. 
Note that the existence of an edge in dependency graph depends on the probability that a UT is requesting a content
and another UT has already cached that content\footnote{This edge has no relationship with the actual physical 
link between two UTs.}. Therefore,
this is a non-deterministic graph with some probability 
for the existence of each edge between the two vertices. 
In this non-deterministic
dependency graph, the probability of existence of edge $(N_{i}, N_{j})$ in $E$ 
is equal to the probability that content $r_i$ requested by $N_i$, is already cached in $N_j$. 
 Therefore, with  LRU or LFU caching policy assumption and using equations
 \eqref{validating_eq} and \eqref{validating_eq2}, we arrive at 
\begin{equation}
 \label{edgeprob}
 \textrm{Pr}[(N_{i},N_{j}) \in E] = \textrm{Pr}[r_i \in C_j]  \ge \frac{r_i^{-s}}{H_{m,s}}.
\end{equation}
 Using Lemma \ref{hproblem} for any $0 < \epsilon <1$, if $h_{\epsilon}$ is 
chosen as $h_{\epsilon} = {\epsilon}^{\frac{1}{1-s}}$, then with a probability of at least 
$1-\epsilon$, any requested content has a popularity index $r_i$ 
less than $h_{\epsilon}$.
This means that with a 
probability { of at least $1-\epsilon$}, 
the edge presence probability in equation \eqref{edgeprob} can be lower 
bounded by $p_{\epsilon}$.
Therefore,  with a probability of at least $1-\epsilon$,
maximum number of vertex-disjoint cycles in our directed dependency graph
$\overrightarrow{G}(V,E)$ can be lower bounded by the maximum number of 
vertex-disjoint cycles in an Erd\H{o}s-R\'eyni random graph $\overrightarrow{G}(n, p_{\epsilon})$ with $n$ vertices and edge 
presence probability $p_{\epsilon}$. Now we can use 
Lemma \ref{lemerdos} and undirected graph ${G}(n, p_{\epsilon}^2)$ 
to find a lower bound on the 
number of vertex disjoint cycles in $\overrightarrow{G}(n, p_{\epsilon})$. This in turn, will give us a lower bound on the number of 
vertex-disjoint cycles in $\overrightarrow{G}(V,E)$.

Note that $G(n,p_{\epsilon}^2)$ is an undirected Erd\H{o}s-R\'eyni random graph on $n$ vertices and 
edge presence probability $p_{\epsilon}^2$. This graph has
 a maximum of $n(n-1)$ undirected 
edges. However, since every undirected
edge 
in this graph exists with a probability of $p_{\epsilon}^2$, the expected value of the number of edges in graph 
$G(n,p_{\epsilon}^2)$
is $n(n-1)p_{\epsilon}^2$. 
This means that if $d$ in Lemma \ref{lemerdos} with $v=n$ is chosen to be an integer such that 
\begin{equation}
 \label{chooces}
 n(n-1)p_{\epsilon}^2 \ge (2d-1)n - 2d^2+d,
\end{equation}
then on average, $G(n,p_{\epsilon}^2)$ will have
{either}
$d$ disjoint cycles
{or $2d-1$ vertices of degree $n-1$}. For the purpose of our paper we can easily
verify that for large enough values of $n$, 
$d^{\star} = \lfloor \frac{np_{\epsilon}^2}{24} \rfloor$ satisfies equation \eqref{chooces} (Notice that the condition 
$24 d^{\star} \le n$ in Lemma \ref{lemerdos} is also met). Therefore {based on Lemma \ref{lemerdos},
the graph 
$G(n,p_{\epsilon}^2)$ either has at least $d^{\star} = \lfloor \frac{np_{\epsilon}^2}{24} \rfloor$ disjoint cycles
or
$2d^{\star}-1=2\lfloor \frac{np_{\epsilon}^2}{24} \rfloor-1$ vertices 
with degree $n-1$. As mentioned before, $\overrightarrow{G}(n, p_{\epsilon})$ essentially contains a 
copy of $G(n,p_{\epsilon}^2)$. Consequently, $\overrightarrow{G}(n, p_{\epsilon})$ either has at least
$d^{\star} = \lfloor \frac{np_{\epsilon}^2}{24} \rfloor$ disjoint cycles
or 
$2d^{\star}-1=2\lfloor \frac{np_{\epsilon}^2}{24} \rfloor-1$ vertices 
with degree $n-1$. The number of vertices in graph $\overrightarrow{G}(n, p_{\epsilon})$ is $n$. Therefore, 
the latter case gives rise to a situation where there are $2d^{\star}-1=2\lfloor \frac{np_{\epsilon}^2}{24} \rfloor-1$ 
vertices which are connected to any 
other vertex in $\overrightarrow{G}(n, p_{\epsilon})$ through undirected edges.
This condition {results in} having a clique of size 
$2\lfloor \frac{np_{\epsilon}^2}{24} \rfloor-1$ in $\overrightarrow{G}(n, p_{\epsilon})$.
}

{In summary, $\overrightarrow{G}(n, p_{\epsilon})$ has either $d^{\star}$ disjoint cycles or it contains 
a clique of size  $2d^{\star}-1$. Hence,
with a probability { of at least $1-\epsilon$}, the dependency graph
$\overrightarrow{G}(V,E)$ on average has either  
$d^{\star}$ disjoint cycles or it contains a clique of size  $2d^{\star}-1$. In either of these cases 
 $d^{\star}=\lfloor \frac{np_{\epsilon}^2}{24} \rfloor$ transmissions can be saved using index coding. 
 This proves the theorem.
}
\end{proof}
\begin{thm}
 \label{tightness}
{\em Index coding through cycle counting and clique partitioning 
 can save $\Theta(n)$ transmissions in a network with $n$ UTs and Zipfian content 
 request distribution with parameter $s>1$.}
\end{thm}
\begin{proof}
 \label{psdfbg}
 In Theorem \ref{indexlimit}, we proved that in a network with Zipfian 
 content request distribution and for a fixed $0< \epsilon <1$, 
 with a probability of at least $1-\epsilon$, the index coding dependency graph 
 either has $d^{\star}=\lfloor \frac{np_{\epsilon}^2}{24} \rfloor$ disjoint cycles 
 or it contains a clique of size $2d^{\star}-1=
 2\lfloor \frac{np_{\epsilon}^2}{24} \rfloor-1$. 
 Consider the following situations, 
 \begin{enumerate}
  \item The dependency graph has 
  $d^{\star}=\lfloor \frac{np_{\epsilon}^2}{24} \rfloor$ disjoint cycles.
  In this case, for a fixed $\epsilon$, $p_{\epsilon}$ is a constant which 
  does not depend 
  on $n$. Hence, cycle counting can result in at least 
  $d^{\star}=\lfloor \frac{np_{\epsilon}^2}{24} \rfloor= \Omega(n)$ transmission savings. 
  This is a lower bound on the number of transmission savings. 
  
  On the other hand,  notice that the maximum number of vertex-disjoint
  cycles in any graph 
  with $n$ vertices cannot be greater than $\frac{n}{2}$ as shown in 
  Figure \ref{fig_ex_proof1}.  Therefore, the maximum number of
  transmission savings using cycle counting  is 
  $\frac{n}{2} = O(n)$. This is an upper bound on the number of 
  saved transmissions. Since the order of upper and lower bounds are the same, it can be concluded that the number of saved 
  transmissions scales as $\Theta(n)$ .

  \item The dependency graph contains a clique of size $\mathrm{k}^{\star}=2d^{\star}-1=
 2\lfloor \frac{np_{\epsilon}^2}{24} \rfloor-1$. Through clique partitioning
 we will  be able to save at least $\mathrm{k}^{\star}-1$
 transmissions by  sending only one transmission.
 Hence, through clique 
 partitioning, we will be able to save at least $\mathrm{k}^{\star}-1 = 
 \Omega(\frac{np_{\epsilon}^2}{12}) =\Omega(n)$ transmissions by sending 
 only one transmission to the UTs forming that specific clique. This is 
 a lower bound on the number of saved transmissions. 
 
 On the other hand, if the dependency graph is
 a perfectly complete graph on $n$ nodes which means that 
 every requested content is available in all other UTs' caches,
 then all the requested transmissions can be satisfied by one transmission
 which is a linear XOR combination
 of all requested contents. 
 Hence, the number of transmission savings is equal to 
 $n-1 =O(n)$. Notice that this is the maximum number of transmission 
 savings since we at least need 1 transmission to satisfy all content 
 requests. This means that the number of transmission savings 
 is upper bounded by $\BigO{n}$. Since the upper and lower order bounds are the same, we conclude that the transmission saving scales as $\Theta(n)$.
 \end{enumerate}
  \end{proof}
  \begin{figure}
\centering
 \begin{tikzpicture}[->,>=stealth',shorten >=1pt,auto,node distance=0.9cm, semithick]
  \node[circle,draw]  (N0) {};
  \node[circle,draw]  (N1) [right of=N0] {};
  \node[circle,draw]  (N2) [below of=N0] {};
  \node[circle,draw]  (N3) [right of=N2] {};
  \node[circle,draw]  (N4) [right of=N3] {};
  \node[circle,draw]  (Nl) [right of=N4] {};
  \node[circle,draw]  (N5) [above of=Nl] {};
  \node[circle,draw]  (N6) [above of=N4] {};
  \path (N0)  [out=-50, in=50] edge node {} (N2);
  \path (N2)  [out=130, in=-130] edge node {} (N0);
  \path (N1)  [out=-50, in=50] edge node {} (N3);
  \path (N3)  [out=130, in=-130] edge node {} (N1);
  \path (N5)  [out=-50, in=50] edge node {} (Nl);
  \path (Nl)  [out=130, in=-130] edge node {} (N5);
  \path (N6)  [out=-50, in=50] edge node {} (N4);
  \path (N4)  [out=130, in=-130] edge node {} (N6);
\end{tikzpicture}
  \caption{Example of a dependency graph on $n=8$ nodes with maximum possible 
number of vertex disjoint cycles $n/2=4$.}
  \label{fig_ex_proof1}
\end{figure}
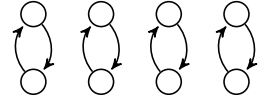

We can further prove that many properties of the dependency graph are independent of the total number of contents 
and only depends on the most popular contents 
in the network. As an example of these properties, we can consider the problem of finding a clique of size 
$k$ in the dependency graph. A clique of size $k$ in the dependency graph has an interesting interpretation since all the requests in this clique can be satisfied with one multicast transmission. 
The following theorem proves that the probability of existence of a clique of size $k$ is 
lower bounded by a value which is 
independent of the 
total number of contents in the network, $m$, and only depends on the popularity index $s$. 
\begin{thm}{\em
 If LRU { or LFU} caching  policy is used and 
 the content request probability is Zipfian distribution, 
 then the probability of finding a set of 
 $k$ UTs $N_b = \{N_{b_1}, N_{b_2},..., N_{b_k} \} \subseteq \mathbb{N}$ for which a
 single linear 
 index code (XOR operation) can be used to transmit the requested content $r_{b_i}$ to $N_{b_i}$ for $1 \le  i \le k$ 
 can be lower bounded by a value that {with a probability close to one} is
 independent of the total number of contents in the network. 
 }\label{thm_clique}
\end{thm}
\begin{proof}
The probability that a specific set of UTs $\{N_{b_1}, N_{b_2},..., N_{b_k}\}$ form a 
clique of size $k$ is 
\begin{eqnarray}
\label{pb1b2bk_initial}
 P_{b_1,b_2,...,b_k} = \textrm{Pr}[r_{b_i} \in C_{b_j} ~\textrm{for}~ 1 \le \forall i,j \le k, j \neq i ~].
\end{eqnarray}
Assuming that the UTs are requesting contents independently of each other, this probability can be simplified as 
\begin{eqnarray}
\label{pb1b2bk_indep}
 P_{b_1,b_2,...,b_k} = \prod_{i=1}^{k}\prod_{j=1,j \neq i}^{k}\textrm{Pr}[r_{b_i} \in C_{b_j}]. 
\end{eqnarray}
Using equations \eqref{validating_eq} and \eqref{validating_eq2},  we arrive at
\begin{eqnarray}
\label{pbicdjkghj_simple}
 \textrm{Pr}[r_{b_i} \in C_{b_j}] \ge \frac{{r_{b_i}}^{-s}}{H_{m,s}}.
\end{eqnarray}
Equation \eqref{pb1b2bk_indep} can  be lower bounded as 
\begin{eqnarray}
\label{pb1b2bk_indep2}
 P_{b_1,b_2,...,b_k} \ge \prod_{i=1}^{k}\left(\frac{{r_{b_i}}^{-s}}{H_{m,s}} \right)^{k-1}.
\end{eqnarray}
The probability to have a clique of size $k$ is computed by considering all $\binom{n}{k}$
groups of $k$ UTs. Hence, the probability of having a clique of size $k$ denoted by $P_k$ is given by
\begin{align}
\label{P_total1}
P_k &= \sum_{b_1,b_2,...,b_k \subseteq N}P_{b_1,b_2,...,b_k}  
  \ge \sum_{b_1,b_2,...,b_k \subseteq N} \prod_{i=1}^{k}\left( 
\frac{{r_{b_i}}^{-s}}{H_{m,s}} \right)^{k-1} \nonumber \\
&= \frac{\sum_{b_1,b_2,...,b_k \subseteq N} \prod_{i=1}^{k}r_{b_i}^{-s (k-1)}}{H_{m,s}^{k-1}}.
\end{align}
In order to simplify this expression, we use the \emph{elementary symmetric polynomial} notation. If we have 
a vector 
$V_n =( v_1,v_2,...,v_n)$ of 
length $n$, then the $k$-th degree elementary symmetric polynomial of these variables is denoted
as
\begin{eqnarray}
\label{elemnt}
\sigma_{k}({V_n})=\sigma_{k}(v_1,...,v_n) = \sum_{1\leq i_1<i_2<..<i_k\leq n}v_{i_1}...v_{i_k}. 
\end{eqnarray} 
Using this notation and by defining $ Y_n \triangleq
(r_{1}^{-s  (k-1)},r_{2}^{-s  (k-1)},...,r_{n}^{-s (k-1)})$,
we have $ P_k \ge  \frac{\sigma_{k}({Y_n})}{H_{m,s}^{k-1}}$.
Since the content request probability follows a Zipfian distribution,  we have 
$\textrm{Pr}[r_j \le h_{\epsilon}] = \frac{H_{h_{\epsilon},s}}{H_{m,s}}$.
Therefore, for a specific group of UTs 
$N_{b_1}, N_{b_2},..., N_{b_k}$, the probability that they all request
contents from the top $h_{\epsilon}$ most popular 
contents is given by   
\begin{eqnarray}
 \label{pr1r2}
 \textrm{Pr}[r_{b_1} \le h,...,r_{b_k} \le h_{\epsilon}]
 = \prod_{j=1}^{k} \textrm{Pr}[ r_{b_j} \le h_{\epsilon}] =
 \left(\frac{H_{h_{\epsilon},s}}{H_{m,s}} \right)^k.
\end{eqnarray}
We have already proved in Lemma \ref{hproblem}  that for large values of $m$ and
 $h_{\epsilon}={\epsilon}^{\frac{1}{1-s}}$, the 
ratio $\frac{H_{h_{\epsilon},s}}{H_{m,s}}$ is greater than $1 - \epsilon$.
Besides this, the fact that $n$ is most likely much 
larger than
$k$, means that with a very high probability, 
for each set of UTs $\{ N_{b_1}, N_{b_2},..., N_{b_k}\} $, the 
requests come only from the $h_{\epsilon}$ most popular contents. This 
implies that with a high probability,
$\sigma_{k}({Y_n}) \ge \binom{n}{k} h_{\epsilon}^{-ks(k-1)}$. 
Also, notice that  
$ H_{m,s} <  \zeta(s) < \infty$.
Therefore, with a probability close to one, $P_k$ can be lower bounded as 
\begin{equation}
 \label{P_lower3}
 P_k \ge \binom{n}{k} \left( \frac{h_{\epsilon}^{-ks}}{\zeta(s)} \right)^{k-1}. 
 \end{equation}
This lower bound does not depend on $m$ and only depends on $n, \epsilon, 
s$ and $k$. 
\end{proof}
Theorem \ref{thm_clique} states that regardless of the number of contents in the 
network, there is always a constant lower bound for the probability of finding a clique of size $k$. 
The result hints the potential use of linear index coding in these networks. In the next section, 
we will prove that linear 
index coding can indeed be very useful and can be used to construct codes acheiving order optimal capacity gains. 

\begin{rem}{\em
 The above capacity improvement is found for a traditional single hop index coding scenario. 
 For our proposed multihop 
 setup, similar gains still hold. In our proposed setting, we consider  communications for a small number 
 of hops and therefore multihop communication can only affect the capacity gain by a constant factor and the 
 order bound results will not be affected.
} \label{rem_hop}
\end{rem}

\section{Heuristics acheiving order optimal capacity}
\label{graphcoloringsec}

Both optimal and approximate solutions \cite{bar2011index,langberg2011hardness} for the general index coding problem are  NP-hard problems. 
Some efficient heuristic algorithms for the index coding problem were proposed \cite{chaudhry2008efficient} which can provide near optimal 
solutions. In some of these heuristic algorithms, the authors reduce the index coding 
problem to the graph coloring problem. 

Notice that every clique in the dependency graph of a specific index coding problem, can be satisfied with only one
transmission which is a linear combination of all contents requested by the UTs corresponding to the clique. Therefore, solving
the clique 
partitioning problem, which is the problem of finding a clique cover of minimum size for a graph \cite{garey2002computers}, 
yields a simple linear  index coding solution. The minimum number of cliques required to cover a graph can be regarded as 
an upper bound on the minimum number of index codes required to satisfy the UTs. 
Index coding rate is defined as the minimum number of required index codes to satisfy all the UTs.
Since lower index coding rates translate into higher values of transmission 
savings (or index coding gains)\footnote{In a dependency graph of $n$ UTs with the index coding rate of $\chi$, the number of
saved transmissions, $n-\chi$, is called the index coding gain.}
 as discussed in \cite{chaudhry2011complementary}, 
the number of transmission savings found in the clique partitioning problem is in fact a lower bound on the total number of 
transmission savings found from the optimal index coding scheme (or the optimal index coding gain).

On the other hand, solving the clique partitioning 
problem for any graph $G(V,E)$ is equivalent to solving the graph coloring problem for the complement graph 
$\bar{G}(V,\bar{E})$ which is a graph on the same set of vertices $V$ but containing only the edges that are not present 
in $E$. This is true because every clique in the dependency graph, gives rise to an independent set
in the complement graph. Therefore, 
if we have a clique partitioning of size $\chi$ in the dependency graph, we have $\chi$ distinct independent sets in the
complement graph. 
In other words, the chromatic number of the complement graph is $\chi$. 

The above argument allows us to use the rich 
literature on the chromatic number of graphs to study the index coding problem. In fact, any graph coloring algorithm running over 
the conflict graph can be directly used to obtain an achievable index coding rate. If running such an algorithm over the conflict 
graph results in a coloring of size $\chi$, this coloring gives rise to a clique cover of size $\chi$ in the dependency graph and 
an index coding of rate $\chi$ with index coding gain of $n-\chi$ which is a lower bound for the total number of transmission 
savings using the optimal index code\footnote{Notice that since the optimal index coding rate is 
upper bounded by the size of the minimum clique cover (which is equal to the chromatic number of the conflict graph), the value of 
transmission savings that we can achieve using the optimal index code is  lower bounded by $n-\chi$.}. 
Therefore, considering the chromatic number of the conflict graph, we can find a lower bound on the asymptotic index coding
gain. To do so, we use the 
following theorem from \cite{bollobas1988chromatic},
\begin{thm}{\em
 For a fixed probability $p$, $0 <p<1$, almost every random graph $G(n,p)$ (a graph with $n$ UTs and the edge
 presence 
 probability of $p$) has chromatic number, 
 \begin{equation}
 \label{chromnil}
 \chi_{G(n,p)} = -\left(\frac{1}{2}+ o(1)\right) \log (1-p) \frac{n}{\log n}
 \end{equation}
}\label{thmgood}
\end{thm}
{We will now use Theorem \ref{thmgood} and the designed undirected graph $G(n,p_{\epsilon}^2)$ to find the number of transmission 
savings using a graph coloring based heuristic in our network.}
\begin{thm}{\em 
Using a graph coloring algorithm, in a network with $n$ UTs almost surely gives us a linear index code with gain
\begin{equation}
\label{gorbec}
 l = \Theta \left(n + \left(\frac{1}{2}+ o(1)\right) \frac{n}{\log n} \log p_{\epsilon}^2 \right).
 \end{equation}
}\label{thmgain}
\end{thm}
\begin{proof}
Assume that a helper is serving $n$ UTs where $n$ is a large number. As 
discussed in Theorem \ref{thmgood}, the index coding gain
is lower bounded by $n-\chi$ where $\chi$ is the chromatic number of the conflict graph. However, notice that on average 
the chromatic
number of our non-deterministic conflict graph is upper bounded by the chromatic number of an undirected
random graph with edge existence 
probability of $1-p_{\epsilon}^2$. 
{To prove this, notice that in the 
dependency graph, the probability of edge existence between two vertices is at 
least $p_{\epsilon}$ which implies that the number of edges in the dependency
graph is on average greater than or equal to the number 
of edges in a directed Erdos-Reyni random graph
$\overrightarrow{G}(n,p_{\epsilon})$. However,  we know 
that the number of edges in $\overrightarrow{G}(n,p_{\epsilon})$ is at least equal
to the number of edges in an undirected 
Erdos-Reyni random graph $G(n,p_{\epsilon}^2)$. Therefore, the conflict
graph which is the complement of dependency graph, on average has less edges 
compared to a random graph with edge
existence probability of $1-p_{\epsilon}^2$ and consequently, its chromatic
number cannot be greater than the chromatic number of $G(n,1-p_{\epsilon}^2)$.
}
Given these facts, the index coding gain is lower bounded by 
$n-\chi_{G(n,1-p_{\epsilon}^2)}$. 
Since $1-p_{\epsilon}^2$ is fixed, Theorem \ref{thmgood} implies that the chromatic number of the conflict graph is equal to 
\begin{equation}
\label{gorbec_new}
 \chi_{G(n,1-p_{\epsilon}^2)} =   -\left(\frac{1}{2}+ o(1)\right) \log p_{\epsilon}^2 ~ \frac{n}{\log n}
 \end{equation}
 This proves that the index coding gain is lower bounded by $\Omega (n \times \{1 + \left(\frac{1}{2}+ o(1)\right) 
 \frac{1}{\log n} \log p_{\epsilon}^2 \})$  which
asymptotically tends to $n$. However, the maximum index coding gain of $n$ UTs is also $n$. Therefore, this coding 
gain is also a tight bound. 
\end{proof}
\begin{rem}{\em
 {Theorem \ref{thmgain} presents the index coding  gain using a graph coloring algorithm which only  
 counts the number of cliques in the dependency graph. The gain in 
 Theorem \ref{indexlimit} 
 counts the number of disjoint cycles in the dependency graph. Theorem
 \ref{tightness} proves that
 index coding  gain
 in Theorem \ref{indexlimit} is $\Theta(n)$ which means that it is order optimal. 
 Theorem \ref{thmgain} is also proving the same result. Therefore, a graph 
 coloring algorithm can acheive order optimal 
 capacity gains.}
 }\label{compare_rem}
\end{rem}

\begin{algorithm}
    \caption{Multihop Caching-Aided Coded Multicasting}
    \label{alg:1}
    \begin{algorithmic}[1]
        \Procedure{Cache Placement}{}
        \State Use greedy cache placement algorithm in \cite{golrezaei2012femtocaching} 
        \Statex ~~~~to populate helper caches.
        \EndProcedure 
	\Statex
        \Procedure{Helper Assignment}{}
        \State Use a { greedy algorithm} to assign the closest helper   
        \Statex ~~~~which has the requested content, to the UT.
        \EndProcedure
   	\Statex
        \Procedure{Content Delivery}{}
        \State { Form} the dependency graph and conflict graph.
        \State { Color} the conflict graph by a { graph coloring}
        \Statex~~~~~ heuristic to find {independent sets} in dependency graph. 
        \For{ every independet set ${S}$ found} 
        \State Helper { multicasts} coded XOR of requested 
        \Statex~~~~~~~~ contents \(\oplus_{k\in{S}}~ M_{r_k}\). 
        \EndFor 
        \EndProcedure
	\Statex       
        \Procedure{Content Decoding}{}
        \State UT $j$ XORs some of it's cached contents  with the
	\Statex ~~~~ received coded content \(\oplus_{k\in{S}}~ M_{r_k}\) to decode $M_{r_j}$
        \EndProcedure
    \end{algorithmic}
\end{algorithm}

\begin{rem}{\em  
 We have shown our proposed algorithm in pseudo-code in Algorithm \ref{alg:1}.
 As we mentioned earlier, in our solution we only focus on content delivery. 
 For cache placement, we assume that the greedy approximate algorithm 
 proposed in \cite{golrezaei2012femtocaching} is used to populate helper 
 caches. For helper assignment, the greedy algorithm is used. In other words,
 we suggest that the closest helper which has the content in its cache
 be assigned to the UT. During the content delivery 
 phase, we use our proposed heuristic for index coding. We use a graph coloring 
 heuristic for the conflict graph and find independent sets in the 
 dependency graph. Then for each independent set only one multicast
 transmission is needed which will be sent by the helper. Our main 
 contribution is to propose better content delivery algorithm compared to the 
 baseline. In the decoding phase, UTs use their cached contents and the 
 received content to decode their desired contents.
 }\label{rem_alg}
\end{rem}

\section{Simulations}
\label{discuss}
In this section, we will show our simulation results.
To show the performance of our coding technique,
we have plotted the simulation results for five different sets of parameters in 
Figure \ref{fig:sim1de}.
In this simulation, we assume that index coding is done in the packet level. 
We plotted the 
average packets sent in each transmission. We have assumed that the UTs are requesting contents based on a Poisson 
distribution with an 
average rate of $\frac{1}{n}$. This way we can assure that the average total request rate is one and we 
are efficiently using the time resource
without generating unstable queues. The optimum solution is an NP-hard problem.
However, 
we used a very simple heuristic algorithm to count the number of cliques and cycles of maximum size 4. 
Even with this simple algorithm, we were able to 
show that the index coding can double the average number of packets per transmission in each time slot for certain values
of the Zipfian parameter. Clearly, optimal index coding 
or more sophisticated algorithms can achieve better results compared to what we obtained by our simple algorithm. 

For small values of $s$, the content request distribution is close to uniform, the dependency graph is very sparse and there is little benefit of using index coding. For large values of $s$, most UTs are requesting similar contents which results in broadcasting the same content to all nodes which is equal to one content per transmission. 
The main benefit of
index coding happens for  values of $s$ between 0.5 and 2 which is usually the 
case in practical networks. Note that a wireless distributed caching system with no index coding, will always have 
one content per transmission. 
%
\begin{figure}[http]
    \center
      \includegraphics[scale=0.5,angle=0]{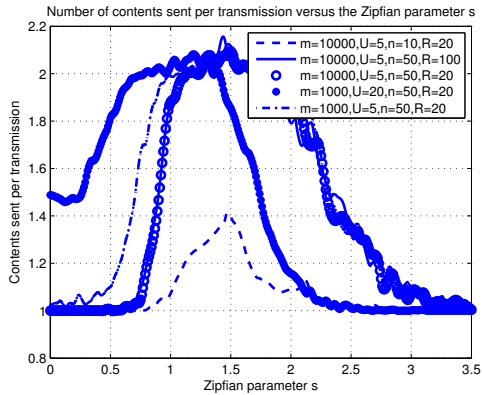}
      \caption{{Average number of contents transmitted in each time slot when using index coding as a function of
      the Zipfian parameter $s$ using a suboptimal search.}}
    \label{fig:sim1de}
\end{figure}

Figure \ref{fig_compare3} compares the average number of requests satisfied 
per time slot by a helper between our proposed scheme and the baseline approach. 
For this simulation, a Zipfian content request distribution with parameter $s=2$
is assumed. The cell radius is 400 meters, D2D transmission range is assumed to 
be 100 meters and 1000 UTs are considered in this figure similar to the simulations 
in \cite{golrezaei2012femtocaching} and Figure \ref{fig_multihop}.
Notice that with the baseline approach at least 27 helpers are required to 
cover the entire network  and 24, 20, 14 and 10 helpers can cover 97\%, 93\%,
78\% and 62\% of the network, respectively. Using multihop D2D 
communications with a maximum of three hops, 4 helpers can cover the entire network 
while 3 helpers can cover 95\% of the network nodes. 
We have shown that even with our very simple
heuristic algorithm, 
multihop D2D can significantly improve the helper utilization ratio. Note that 
in baseline approach, each helper can at most transmit one content per
transmission, however, our approach can satisfy more than one request per
transmission by taking advantage of the side information that is
stored in nodes' caches. As the number of requests per user increases, there are 
more possibilities of creation of cliques of large size which results in increasing
the efficiency of helper nodes. 
The simulation is carried for content request probability of up to 0.3 since 
realistically, no more than one third of nodes at any given time,
request contents in the network.

\begin{figure}[http]
    \center
      \includegraphics[scale=0.5,angle=0]{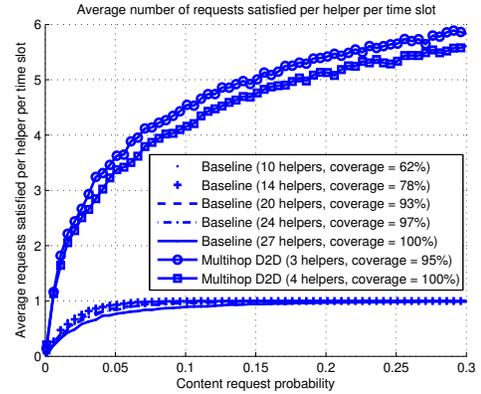}
      \caption{ Comparing our proposed and baseline solution on average 
      number of requests satisfied per time slot per helper versus 
      the content request probability.}
    \label{fig_compare3}
\end{figure}

\section{Conclusion}
\label{conclude} 

An efficient order optimal content delivery approach is proposed for future wireless cellular systems. We take advantage of femtocaches \cite{golrezaei2012femtocaching} and multihop communication using index coding. 
We proved that using index coding is very efficient technique by taking advantage of side information stored in UTs. 
Further, it was shown that under Zipfian content distribution, linear index coding could be order optimal. A heuristic  graph coloring algorithm is proposed that
 achieves  order optimal capacity bounds. Our simulation result demonstrates the gains that can be achieved with this approach. 

Our proposed scheme improves the efficiency of femtocaches by taking advantage
of multihop communications and index coding. One challenge of multihop 
communication is connectivity when nodes are moving fast. One solution is 
to divide the contents into equal smaller chunks and treat each chuck as a
content. Another challenge is the additional delay imposed by using multihop
communications.  However, it is not readily clear that baseline approach
provides lower delays. The 
reason is the fact that by using multihop communication, larger portion of
the cell can be covered when the same number of femtocaches are used in both
approaches as shown in Figure \ref{fig_multihop}. A more detailed comparison 
of mobility and delay in single hop and multihop communication is the subject
of future study.  Security, overhead, routing are other issues that should be
investigated in future works. 
  
\bibliographystyle{plain}
\bibliography{All-Papers}

\begin{IEEEbiography}
[{\includegraphics[width=1in,height=1.25in,clip,keepaspectratio]{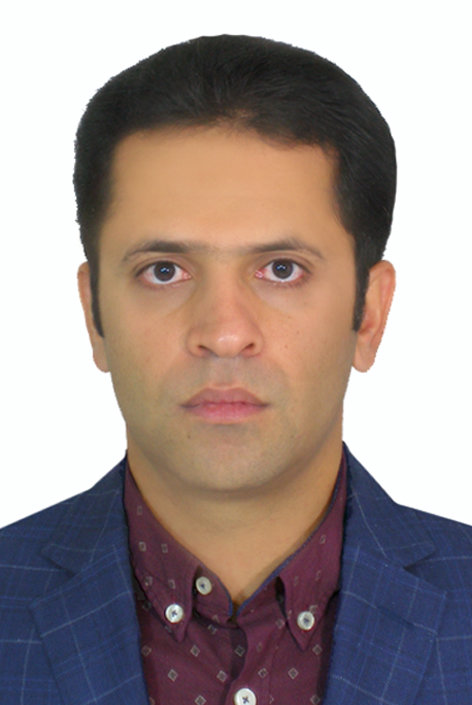}}]%
{Mohsen Karimzadeh Kiskani} received his bachelors degree in mechanical engineering from
Sharif University of Technology in 2008. He got his Masters degree
in electrical engineering from Sharif University of Technology in 2010. 
He got a Masters degree in computer science from University of California Santa Cruz in 2016.
He is currently a PhD candidate in electrical engineering department at University 
of California, Santa Cruz. His main areas of interest include wireless communications and 
information theory. He is also interested in the complexity study of Constraint Satisfaction 
Problems (CSP) in computer science.
\end{IEEEbiography}
\begin{IEEEbiography}
[{\includegraphics[width=1in,height=1in,clip,keepaspectratio]{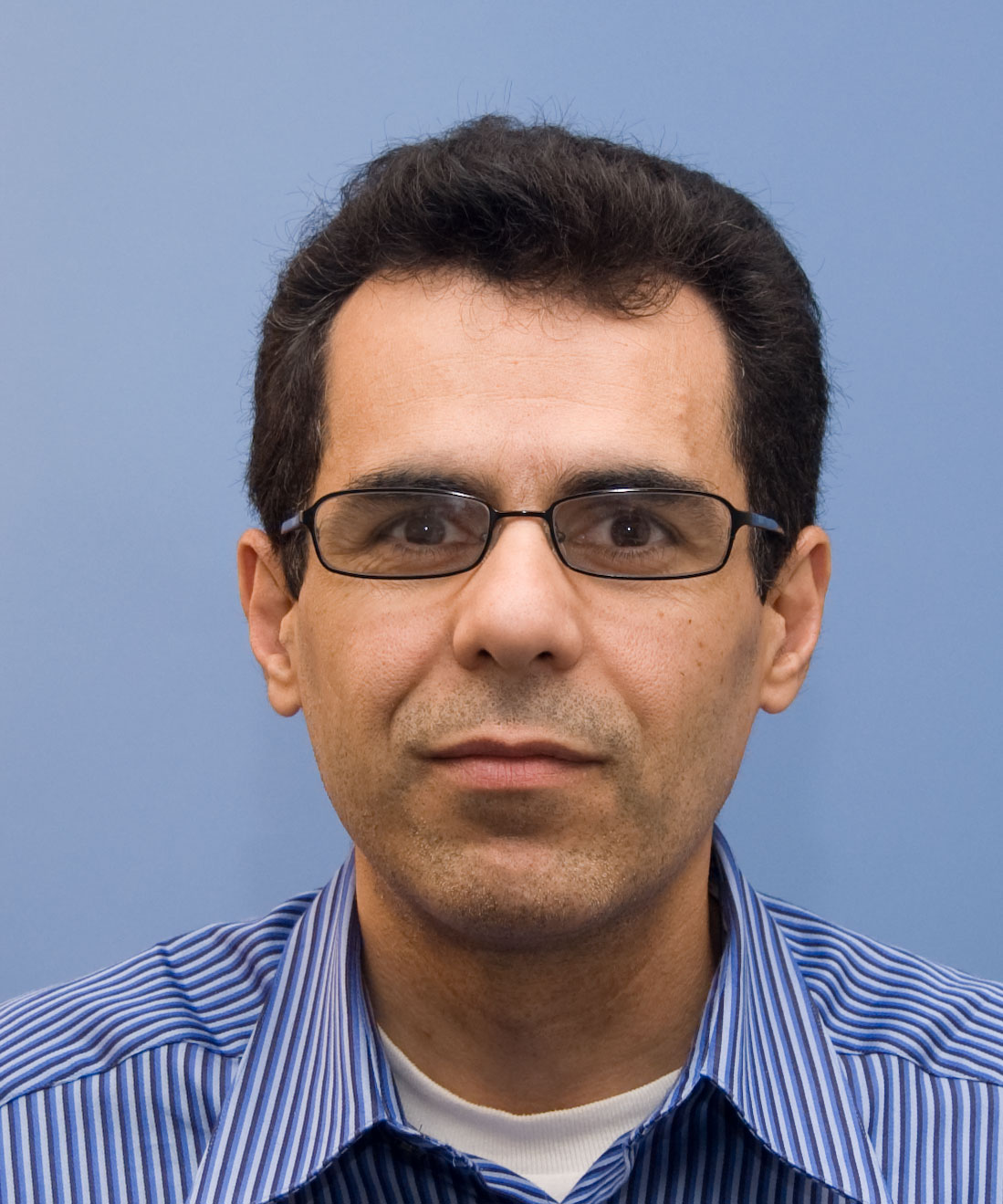}}]%
{Hamid Sadjadpour} (S’94-M’95-SM’00) received his B.S. and M.S. degrees from Sharif University of Technology and Ph.D. degree from University of Southern
California, respectively. After graduation, he joined AT\&T as a member of technical staff and finally Principal member of technical staff until 2001. In fall 2001, he joined University of California, Santa Cruz (UCSC) where he is now a Professor. Dr. Sadjadpour
has served as technical program committee member and chair in numerous conferences. He has published more than 170 publications
and has awarded 17 patents. His research interests are in the general area of wireless communications and networks. He is the co-recipient of best paper
awards at 2007 International Symposium on Performance Evaluation of Computer and Telecommunication Systems (SPECTS), 2008 Military Communications (MILCOM) 
conference, and 2010 European Wireless Conference. 
\end{IEEEbiography}
\end{document}